\documentclass[aps,pra,reprint,superscriptaddress,amsmath,amssymb,amsfonts,nofootinbib,floatfix]{revtex4-2}

\usepackage{amsthm}
\usepackage{mathtools}   
\usepackage{graphicx}
\usepackage{microtype}

\newtheorem{theorem}{Theorem}[section]
\newtheorem{proposition}[theorem]{Proposition}
\newtheorem{lemma}[theorem]{Lemma}
\newtheorem{corollary}[theorem]{Corollary}
\newtheorem{definition}[theorem]{Definition}
\newtheorem{remark}[theorem]{Remark}

\newcommand{\hp}{\hat p}
\newcommand{\hx}{\hat x}
\newcommand{\hK}{\hat K}
\newcommand{\hH}{\hat H}
\newcommand{\hPi}{\hat\Pi}
\newcommand{\ave}[1]{\langle #1\rangle}

\begin{document}

\title{Uncertainty Relations and Quantum Speed Limits for Lattice-Regularized Fractional Kinetic Operators}

\author{M. W. AlMasri}
\email{mwalmasri2003@gmail.com}
\affiliation{Wilczek Quantum Center, School of Physics and Astronomy, Shanghai Jiao Tong University, Minhang, Shanghai 200240, China}

\begin{abstract}
We study a ``hybrid'' kinetic operator on $L^2(\mathbb R)$ that combines a lattice-regularized momentum, with lattice constant $a$, and a fractional (L\'evy-type) power $\alpha\in(1,2]$. It is a bounded Fourier multiplier whose symbol is a fractional power of the central-difference Laplacian symbol; equivalently, it is a long-range hopping model with algebraically decaying hops. We prove that the kinetic operator is bounded and self-adjoint with purely absolutely continuous spectrum $[0,E_{\max}]$ (for every $\alpha\in(0,2]$), and that $\hat K+V$ is self-adjoint on $\mathcal D(V)$. The main new technical result is an exact position--momentum commutator with domain control for the singular multiplier it produces when $\alpha\in(1,2)$, and the resulting Robertson-type uncertainty relation together with its joint weak-lattice, weak-fractionality expansion. For $\alpha=2$ the uncertainty relation reduces to a known lattice result, which we attribute. For quantum speed limits we recall the bandwidth bound $\tau_\perp\ge\pi\hbar/E_{\max}$, a direct consequence of the Mandelstam--Tamm bound, and evaluate the Mandelstam--Tamm and Margolus--Levitin bounds for Gaussian states in closed form, including their crossover as a function of $\alpha$. To leading order the lattice reduces the energy variance of a Gaussian state and therefore raises the Mandelstam--Tamm bound; the effect of the fractional order is controlled by the width of the state relative to a reference momentum $p_0$ through an explicit crossover width $\sigma_*$. For $\hat H=\hat K-F\hat x$ we show that $\ave{\hx}(t)$ is exactly periodic for every state, and discuss ion-chain realizations.
\end{abstract}

\maketitle

\section{Introduction}
Standard quantum mechanics (SQM) rests on a Hilbert space with canonical commutation relations, a local kinetic operator, and unitary dynamics generated by a self-adjoint Hamiltonian \cite{Dirac1958,VonNeumann1955,Mackey1963}. Two simple modifications of the free kinetic term are well studied. Fractional quantum mechanics (FQM) replaces $p^2/2m$ by a L\'evy-type term $D_\alpha|p|^\alpha$ and introduces spatial non-locality \cite{Laskin2000,Laskin2002,Samko1993}. Lattice regularization replaces $p$ by a finite-difference momentum, producing a bounded, periodic dispersion. Fractional powers of discrete Laplacians arise from long-range lattice couplings \cite{Tarasov2006,Kirkpatrick2013,LaskinZaslavsky2006} and have been analyzed mathematically \cite{Ciaurri2018}. On the uncertainty side, the lattice ($\alpha=2$) relation $\Delta x\,\Delta p_a\ge\frac\hbar2|\ave{\cos(a\hat p/2\hbar)}|$ and its weak-lattice expansion were obtained by Jizba, Kleinert and Scardigli \cite{Jizba2010}. On the speed-limit side, bounds for spectra of finite width are an active topic \cite{Ness2022,Allan2021}, building on the Mandelstam--Tamm, Margolus--Levitin and related bounds \cite{Mandelstam1945,Margolus1998,Anandan1990,Pfeifer1993,Uffink1993,Levitin2009,Deffner2017}. To our knowledge, position--momentum uncertainty relations for the \emph{fractional} lattice operator, with explicit control of the operator domains required by the singular commutator multiplier, have not been given.

We emphasize that the model is a Fourier multiplier with a periodic symbol acting on the continuum space $L^2(\mathbb R)$, i.e.\ a direct integral of lattice fibers, and not an operator on $\ell^2(a\mathbb Z)$. Section~\ref{sec:lattice_vs_continuum} states precisely which results transfer to a genuine lattice.

Main results:
\begin{enumerate}
\item[(i)] A dimensionally consistent definition of the hybrid kinetic operator with fractional order $\alpha$, reference momentum $p_0$ and lattice constant $a$; boundedness, self-adjointness and spectrum (Theorem~\ref{thm:selfadjoint}).
\item[(ii)] An exact commutator with domain control for $\alpha\in(1,2]$ (Lemma~\ref{lemma:commutator}), the hybrid uncertainty relation (Theorem~\ref{thm:uncertainty}; known for $\alpha=2$ \cite{Jizba2010}, new for $\alpha\ne2$), and its joint weak-lattice, weak-fractionality expansion with a justification of its range of validity (Proposition~\ref{prop:expansion}).
\item[(iii)] Closed-form Gaussian evaluations of the Mandelstam--Tamm (MT) and Margolus--Levitin (ML) bounds, the MT/ML crossover in $\alpha$, and the precise dependence on the width $\sigma/p_0$ (Proposition~\ref{prop:gaussian}). The bandwidth bound $\tau_\perp\ge\pi\hbar/E_{\max}$ (Proposition~\ref{prop:bandwidth}) is recalled for completeness; it follows from MT and Popoviciu's inequality and is a weaker form of known bounded-spectrum results \cite{Ness2022,Allan2021}.
\item[(iv)] An exact periodicity result for $\ave{\hx}(t)$ under a constant force, and ion-chain realizations (Section~\ref{sec:phenomenology}).
\end{enumerate}

\section{Lattice-fractional operators}
\label{sec:framework}
We work in $\mathcal H=L^2(\mathbb R)$ with $\tilde\psi(k)=(2\pi)^{-1/2}\int e^{-ikx}\psi(x)\,dx$, momentum $p=\hbar k$, and $\hat x=i\,d/dk$ in the $k$-representation, so that $\mathcal D(\hat x)=\{\psi:\tilde\psi\in H^1(\mathbb R)\}$. Functions of $\hat p$ act as multiplication by the function of $p=\hbar k$; we write $K(k)\equiv K(\hbar k)$ when no confusion arises. Throughout, $\theta\coloneqq ap/2\hbar=ak/2$.

\begin{definition}[Lattice momentum and hybrid momentum]
Let $a>0$ be a length, $p_0>0$ a reference momentum, $\alpha\in(0,2]$ and $\delta\coloneqq2-\alpha$. Define
\begin{align}
s_a(p)&\coloneqq\frac{2\hbar}{a}\sin\theta,\qquad \theta=\frac{ap}{2\hbar}, \label{eq:lattice_symbol}\\
\Pi_{a,\alpha}(p)&\coloneqq p_0\,\mathrm{sgn}\big(s_a(p)\big)\,\Big|\frac{s_a(p)}{p_0}\Big|^{\alpha/2}, \label{eq:hybrid_symbol}
\end{align}
and the operators $\hat\Pi_{a,\alpha}\coloneqq\Pi_{a,\alpha}(\hat p)$ (bounded and self-adjoint). The hybrid kinetic operator and Hamiltonian are
\begin{equation}
\hK_{a,\alpha}\coloneqq\frac{\hat\Pi_{a,\alpha}^2}{2m}=\frac{p_0^{\delta}}{2m}\,|s_a(\hat p)|^{\alpha},\qquad \hH_{a,\alpha}\coloneqq\hK_{a,\alpha}+V(\hat x),
\label{eq:K_def}
\end{equation}
with $V$ real-valued and measurable, and $\mathcal D(\hH_{a,\alpha})=\mathcal D(V)$. The commutator results of Section~\ref{sec:uncertainty} require $\alpha\in(1,2]$, i.e.\ $\delta\in[0,1)$; the spectral and speed-limit results hold for all $\alpha\in(0,2]$.
\end{definition}

\begin{remark}[Dimensions, parity, role of $p_0$, limits]
$s_a$ has dimension of momentum, and $\Pi_{a,\alpha}$ has dimension of momentum for every $\alpha$ thanks to $p_0$. The prefactor $D_\alpha=p_0^{\delta}/2m$ has the dimensions of energy$/$momentum$^\alpha$ and equals $1/2m$ at $\alpha=2$; it plays the role of Laskin's $D_\alpha$. Its choice gives $p_0$ a clean operational meaning: at $a=0$ all dispersions $D_\alpha|p|^\alpha$ intersect at $|p|=p_0$, where $K=p_0^2/2m$ for every $\alpha$. Comparisons between different $\alpha$ at fixed $D_\alpha=p_0^\delta/2m$ are therefore comparisons of dispersions that agree at $p_0$; statements such as ``fractionality raises or lowers a bound'' are relative to this convention (changing $p_0$ rescales $\Pi$ by the constant factor $(p_0'/p_0)^{\delta/2}$). The real branch in \eqref{eq:hybrid_symbol} makes $\Pi_{a,\alpha}$ odd, real-valued and well defined for all $p$; the sign alternates between Brillouin zones automatically. At $\alpha=2$, $\Pi_{a,2}=s_a$. As $a\to0$, $s_a(p)\to p$ and $\Pi_{a,\alpha}\to p_0\,\mathrm{sgn}(p)|p/p_0|^{\alpha/2}$; as $a\to0$ and $\alpha\to2$ one recovers $p$ and $p^2/2m$.
\end{remark}

\begin{theorem}[Self-adjointness and spectrum]
\label{thm:selfadjoint}
For $a>0$ and $\alpha\in(0,2]$, $\hK_{a,\alpha}$ is a bounded self-adjoint operator on $L^2(\mathbb R)$, with purely absolutely continuous spectrum
\begin{equation}
\sigma(\hK_{a,\alpha})=[0,E_{\max}],\qquad E_{\max}=\frac{p_0^\delta}{2m}\Big(\frac{2\hbar}{a}\Big)^{\alpha}.
\end{equation}
For real $V$, $\hH_{a,\alpha}=\hK_{a,\alpha}+V(\hx)$ is self-adjoint on $\mathcal D(V)$; if $V$ is bounded below so is $\hH_{a,\alpha}$.
\end{theorem}
\begin{proof}
$\hK_{a,\alpha}=\mathcal F^{-1}M_{K}\mathcal F$ with $K(p)=\frac{p_0^\delta}{2m}|s_a(p)|^\alpha$ real, continuous and bounded by $E_{\max}$, so $\hK$ is bounded self-adjoint. The range of $K$ is $[0,E_{\max}]$. $K$ is $2\pi\hbar/a$-periodic and piecewise real-analytic, and $K'(p)$ (see \eqref{eq:velocity}) vanishes only on a discrete set (zeros of $s_a$ and of $\cos\theta$). Since the critical points of $K$ are discrete, $K^{-1}(N)$ is a Lebesgue null set whenever $N$ is null: on each of the (countably many) open monotone pieces $K$ is a real-analytic diffeomorphism onto its image, and the discrete set of critical points contributes a null set. Hence the push-forward of Lebesgue measure under $K$ is absolutely continuous, and so are all spectral measures of $\hK$. A bounded self-adjoint operator added to a self-adjoint operator $V(\hx)$ is self-adjoint on $\mathcal D(V)$ \cite{Kato1966}.
\end{proof}

\begin{remark}[Bounded kinetic energy; no momentum quantization]
The spectrum is a continuous band. The bandwidth $E_{\max}\propto a^{-\alpha}$ caps the kinetic energy and, for free evolution, the energy variance (Proposition~\ref{prop:bandwidth}).
\end{remark}

\paragraph{Non-locality.} Since $K$ is $2\pi\hbar/a$-periodic, $\hK$ is a series of translations by multiples of $a$:
\begin{equation}
\hK\psi(x)=\sum_{n\in\mathbb Z}\kappa_n\,\psi(x-na),\qquad \kappa_n=\frac{a}{2\pi}\int_{-\pi/a}^{\pi/a}K(k)\,e^{inak}\,dk,
\label{eq:hopping}
\end{equation}
with $K(k)\equiv K(\hbar k)$. For $\alpha\in(1,2)$ the coefficients decay as $|\kappa_n|\sim C_\alpha|n|^{-(1+\alpha)}$ (from the $|k|^\alpha$ non-analyticity of $K$ at $k=0$), with a prefactor $C_\alpha$ vanishing like $(2-\alpha)$ as $\alpha\to2$; at $\alpha=2$, $K$ is a trigonometric polynomial and the hopping is exactly nearest-neighbor. For $\alpha>1$, $\sum_n|\kappa_n|<\infty$, so the series converges in operator norm. It is the \emph{kinetic} term that samples the wavefunction at distant points; the potential $V(x)$ acts locally.

\paragraph{Dispersion.} For a plane wave $e^{ikx}$,
\begin{align}
E_{a,\alpha}(k)&=\frac{p_0^\delta}{2m}\Big[\frac{2\hbar^2}{a^2}\big(1-\cos ak\big)\Big]^{\alpha/2},\label{eq:dispersion}\\
v(k)&=\frac1\hbar\frac{\partial E}{\partial k}=\frac{\alpha p_0}{2m}\,\mathrm{sgn}(s_a)\Big|\frac{s_a}{p_0}\Big|^{\alpha-1}\cos\theta,
\label{eq:velocity}
\end{align}
and $|v|\le\frac{\alpha p_0}{2m}\big(\frac{2\hbar}{ap_0}\big)^{\alpha-1}$. For $\alpha=2$, $E=\frac{\hbar^2}{ma^2}(1-\cos ak)$ and the effective mass is $m^*(k)=m/\cos(ak)$, negative for $ak\in(\pi/2,3\pi/2)$ mod $2\pi$.

\begin{table*}[htbp]
\caption{Comparison of limiting cases ($\delta=2-\alpha$, $\theta=ap/2\hbar$).}
\label{tab:frameworks}
\begin{ruledtabular}
\begin{tabular}{lccc}
 & SQM & Lattice ($\alpha=2$) & Fractional ($a\to0$) \\
\hline
Kinetic symbol & $p^2/2m$ & $\frac{2\hbar^2}{ma^2}\sin^2\theta$ & $\frac{p_0^\delta}{2m}|p|^\alpha$ \\
Spectrum of $\hK$ & $[0,\infty)$ & $[0,2\hbar^2/ma^2]$ & $[0,\infty)$ \\
$M(p)=\partial_p\Pi$ & $1$ & $\cos\theta$ & $\frac\alpha2|p/p_0|^{-\delta/2}$ \\
Bound on $\Delta x\Delta\Pi$ & $\hbar/2$ & $\frac\hbar2|\ave{\cos\theta}|$ & $\frac{\alpha\hbar}{4}|\ave{|p/p_0|^{-\delta/2}}|$ \\
\end{tabular}
\end{ruledtabular}
\end{table*}

\section{Generalized uncertainty relations}
\label{sec:uncertainty}
\begin{lemma}[Hybrid commutator]
\label{lemma:commutator}
For $\alpha\in(1,2]$ and $\psi\in\mathcal D(\hx)$, $\hPi\psi\in\mathcal D(\hx)$ and
\begin{equation}
[\hx,\hPi_{a,\alpha}]=i\hbar\,M_{a,\alpha}(\hp),\qquad M_{a,\alpha}(p)\coloneqq\frac{\alpha}{2}\Big|\frac{s_a(p)}{p_0}\Big|^{-\delta/2}\cos\theta.
\label{eq:hybrid_commutator}
\end{equation}
$M_{a,\alpha}$ is periodic and locally square integrable, but \emph{unbounded} for $\alpha<2$ (it diverges like $|p|^{-\delta/2}$ at the zeros of $s_a$); it is bounded, $M=\cos\theta$, at $\alpha=2$. No $\mathrm{sgn}$ factors appear: $M$ is negative where $\cos\theta<0$. For $\alpha\le1$ the multiplier $M$ is no longer locally square integrable and the commutator holds only at the level of quadratic forms on restricted domains; we do not pursue this case.
\end{lemma}
\begin{proof}
In the $k$-representation, $\hx=i\,d/dk$, so $[\hx,F(k)]\tilde\psi=iF'(k)\tilde\psi$ whenever $F\tilde\psi\in H^1$. Here $F=\Pi_{a,\alpha}$ is bounded, continuous, and absolutely continuous with $F'(k)=\hbar\,\partial_p\Pi=\hbar M(\hbar k)$: indeed $\partial_p[\mathrm{sgn}(s)|s|^{\alpha/2}]=\frac\alpha2|s|^{\alpha/2-1}s_a'(p)$ and $s_a'(p)=\cos\theta$; the factor $p_0^{\delta/2}$ is absorbed by $|s/p_0|$. Because $\alpha>1$, $|M|^2\sim|k|^{-\delta}$ is locally integrable ($\delta<1$) and $M$ is periodic, so $\sum_n\sup_{\text{cell }n}|\tilde\psi|^2\le C\|\tilde\psi\|_{H^1}^2$ gives $F'\tilde\psi\in L^2$. Hence $F\tilde\psi\in H^1$ with $(F\tilde\psi)'=F'\tilde\psi+F\tilde\psi'$, and the commutator identity follows. Details are in Appendix~\ref{app:commutator}.
\end{proof}

\begin{theorem}[Hybrid uncertainty principle]
\label{thm:uncertainty}
For every normalized $\psi\in\mathcal D(\hx)$, $\ave{M_{a,\alpha}(\hp)}$ is finite (it follows from the estimates in the proof of Lemma~\ref{lemma:commutator}, since $M\in L^2_{\rm loc}$ and periodic) and
\begin{equation}
\Delta x\,\Delta\Pi_{a,\alpha}\ \ge\ \frac\hbar2\,\big|\ave{M_{a,\alpha}(\hp)}\big|,
\label{eq:hybrid_uncertainty}
\end{equation}
where $\Delta\Pi_{a,\alpha}$ is the standard deviation of $\hat\Pi_{a,\alpha}$. The inequality is covariant under $p_0\to p_0'$: both sides are multiplied by $(p_0'/p_0)^{\delta/2}$. For $\alpha=2$ it reads $\Delta x\,\Delta p_a\ge\frac\hbar2|\ave{\cos(a\hp/2\hbar)}|$, which is the lattice uncertainty relation of Jizba, Kleinert and Scardigli \cite{Jizba2010} (under their identification $\epsilon=a/2$). Moreover, using $\cos\theta\ge1-\theta^2/2$, if $\ave{\hp^2}<\infty$,
\begin{equation}
\Delta x\,\Delta p_a\ \ge\ \frac\hbar2\Big[1-\frac{a^2\ave{\hp^2}}{8\hbar^2}\Big]\quad\text{whenever the bracket is positive,}
\label{eq:alpha2_bound}
\end{equation}
which is also contained in \cite{Jizba2010}.
\end{theorem}
\begin{proof}
Apply the Robertson--Schr\"odinger inequality to $\hx$ and $\hat\Pi$ (both operators map into the required domains by Lemma~\ref{lemma:commutator}, and $\hat\Pi$ is bounded), drop the non-negative anticommutator term, and insert \eqref{eq:hybrid_commutator}. The covariance statement follows from $M\propto p_0^{\delta/2}$ and $\Pi\propto p_0^{\delta/2}$.
\end{proof}

\begin{remark}[Tightness]
The bound is asymptotically saturated by narrow-momentum wavepackets centered at $p_*$: $\Delta x\,\Delta\Pi\to\frac\hbar2|M(p_*)|$ as $\Delta p\to0$ (Gaussian wavepackets saturate $\Delta x\Delta p\ge\hbar/2$ and $\Delta\Pi\simeq|M(p_*)|\Delta p$). No normalizable state saturates it exactly: equality in the Robertson inequality requires $(\hat\Pi-\ave\Pi)\psi=i\lambda(\hx-\ave x)\psi$, i.e.\ $\tilde\psi'=-\lambda^{-1}(F-\ave\Pi)\tilde\psi+i\ave x\,\tilde\psi$ in the $k$-representation; since $F$ is bounded, periodic and odd (zero band-average), $|\tilde\psi|$ is bounded below by $c\,e^{\ave\Pi k/\lambda}$, which cannot be square integrable. Finally, $\ave M=0$ does not imply a small product $\Delta x\Delta\Pi$; it only means that \eqref{eq:hybrid_uncertainty} is silent.
\end{remark}

\begin{proposition}[Weak-lattice, weak-fractionality expansion]
\label{prop:expansion}
Let $\theta=ap/2\hbar$ and $L=\ln(|p|/p_0)$. To first order in $\delta$ (exactly in $a$, and away from the zeros $p_n=2\pi n\hbar/a$ of $s_a$),
\begin{equation}
M_{a,\alpha}(p)=\cos\theta\Big[1-\frac\delta2\Big(1+L+\ln\Big|\frac{\sin\theta}{\theta}\Big|\Big)\Big]+O(\delta^2),
\end{equation}
and for a state whose momentum distribution is concentrated at $|p|\ll\hbar/a$ (precisely: $\ave{\hp^2}$ and $\ave{\hp^2L}$ finite and the contribution of $|p|\gtrsim\hbar/a$ negligible, see below),
\begin{equation}
\begin{aligned}
\ave{M}={}&1-\frac{a^2\ave{\hp^2}}{8\hbar^2}-\frac\delta2\big(1+\ave{L}\big)\\
&+\frac{\delta a^2}{\hbar^2}\Big(\frac{\ave{\hp^2}}{12}+\frac{\ave{\hp^2L}}{16}\Big)+O(a^4,\delta^2),
\end{aligned}
\label{eq:uncertainty_expanded}
\end{equation}
so that $\Delta x\Delta\Pi\ge\frac\hbar2|\ave M|$ with $\ave M$ given above. For a Gaussian momentum distribution of width $\sigma=\Delta p$, $\ave{L}=\ln(\sigma\sqrt2/p_0)+\tfrac12\psi(\tfrac12)$, with $\psi$ the digamma function.
\end{proposition}
\noindent\emph{Range of validity.} (i) The $O(\delta)$ expansion is not uniform near the zone centers $p_n$, where $M$ diverges for $\alpha<2$ and the remainder involves $(\ln|s_a/p_0|)^2$; the pointwise statement above is therefore local. (ii) Integrating term by term requires $\ave{\hp^2}<\infty$ and $\ave{\hp^2L}<\infty$ (and $\ave{|L|}<\infty$); these are not guaranteed for general $\psi\in\mathcal D(\hx)$ and are assumed. (iii) For a Gaussian distribution of width $\sigma$, the integrals of the exact $M$ and of its expansion over $|p|\ge c\hbar/a$ are both bounded by $O(e^{-c'\hbar^2/(a\sigma)^2})$, because $M$ is locally integrable and periodic while the Gaussian tail is super-exponentially small; dominated convergence on $|p|\le c\hbar/a$ (where the remainders are dominated by an integrable multiple of $\theta^4+\delta^2L^2$) then shows that \eqref{eq:uncertainty_expanded} is a legitimate joint asymptotic expansion for $a\sigma/\hbar\to0$, $\delta\to0$.

The lattice therefore \emph{lowers} the right-hand side of the bound by $a^2\ave{\hp^2}/8\hbar^2$ at $\alpha=2$. The fractional order lowers it when $\ave{\ln(|p|/p_0)}>-1$ (typical momenta above $p_0/e$) and raises it otherwise. The left-hand side $\Delta\Pi$ also depends on $a$ and $\alpha$, so these statements concern the bound's right-hand side only. The derivation is in Appendix~\ref{app:expansion}.

\begin{corollary}[Limiting cases]
\label{cor:limits}
\begin{enumerate}
\item[(a)] \textbf{SQM} ($a\to0,\alpha\to2$): $\Delta x\Delta p\ge\hbar/2$.
\item[(b)] \textbf{Lattice QM} ($\alpha=2$): $\Delta x\Delta p_a\ge\frac\hbar2|\ave{\cos(a\hp/2\hbar)}|$ ( \cite{Jizba2010}).
\item[(c)] \textbf{FQM} ($a\to0$): $\Delta x\Delta\Pi_\alpha\ge\frac{\alpha\hbar}{4}\big|\ave{|\hp/p_0|^{-\delta/2}}\big|\approx\frac\hbar2\big[1-\frac\delta2(1+\ave{\ln(|\hp|/p_0)})\big]$.
\end{enumerate}
\end{corollary}

\begin{remark}[No state-independent bound and no minimal length]
The right-hand side of \eqref{eq:hybrid_uncertainty} is state dependent and vanishes for states with $\ave{M}=0$ (e.g.\ $\ave{\cos\theta}=0$), so no state-independent lower bound on $\Delta x$ follows. For states supported in the first Brillouin zone $|p|<\pi\hbar/a$ one has $\cos\theta>0$ and the bound is strictly positive. This absence of a lower bound near the zone boundary, and the contrast with minimal-length generalized uncertainty principles \cite{Kempf1995}, were discussed for $\alpha=2$ in \cite{Jizba2010}. Minimal-length GUPs require a commutator correction that grows with momentum, whereas $M$ is bounded and oscillatory (at $\alpha=2$) or periodic with integrable singularities. We make no minimal-length claim.
\end{remark}

\begin{proposition}[Velocity operator and Ehrenfest relations]
\label{prop:ehrenfest}
Since $\hK=\hat\Pi^2/2m$, the Leibniz rule gives $[\hx,\hK]=\frac{i\hbar}{m}\hat\Pi M(\hat p)=i\hbar\,v(\hat p)$ with $v$ as in \eqref{eq:velocity}. Assume $\alpha\in(1,2]$, $V\in C^1(\mathbb R)$ with $V'$ bounded (or, more generally, relatively bounded in the sense needed for $V'(\hx)\psi_t\in L^2$), and $\psi_0\in\mathcal D(\hx)\cap\mathcal D(\hat p)\cap\mathcal D(\hH)$, with the evolved state remaining in these domains. Then
\begin{equation}
\frac{d}{dt}\ave{\hx}=\ave{v(\hat p)},\qquad \frac{d}{dt}\ave{\hat p}=-\ave{V'(\hx)},
\end{equation}
where $\hat p$ is the canonical momentum. For $a\to0,\alpha\to2$, $v\to p/m$. The velocity law is independent of $V$; the force law is local because $[V(\hx),\hat p]=i\hbar V'(\hx)$. The position--energy relation is $\Delta x\,\Delta K\ge\frac\hbar2|\ave{v(\hat p)}|$, which is trivial for states with $\ave{v}=0$ (e.g.\ parity-even states).
\end{proposition}

\section{Quantum speed limits}
\label{sec:qsl}
For time-independent $\hH$ and a state with $|\ave{\psi_0|\psi_\tau}|=\cos\Phi$, $\Phi\in[0,\pi/2]$, the Mandelstam--Tamm (MT) bound is $\tau\ge\hbar\Phi/\Delta H$ \cite{Mandelstam1945,Anandan1990,Pfeifer1993,Uffink1993}; if an orthogonal state is reached ($\Phi=\pi/2$) then $\tau_\perp\ge\pi\hbar/2\Delta H$. The Margolus--Levitin (ML) bound is $\tau_\perp\ge\pi\hbar/2\ave{\hH-E_0}$ with $E_0=\inf\sigma(\hH)$ \cite{Margolus1998,Levitin2009}. For reviews see \cite{Deffner2017}; for a recent general treatment see \cite{delCampo2026}. Both orthogonal-state bounds are statements about states that actually become orthogonal; they are vacuous otherwise. With $\Delta H^2=\Delta K^2+\Delta V^2+2\,\mathrm{Cov}(\hK,V)$, where $\mathrm{Cov}(\hK,V)=\tfrac12\ave{\hK V+V\hK}-\ave{\hK}\ave V$, the MT bound for $\hH_{a,\alpha}$ is a direct application on $\mathcal D(V)$ (Theorem~\ref{thm:selfadjoint}). We make no general statement on the sign of the covariance.

\begin{proposition}[Bandwidth bound]
\label{prop:bandwidth}
For free evolution ($V=0$), every state satisfies $\Delta H\le E_{\max}/2$ (Popoviciu's inequality for a variable supported in $[0,E_{\max}]$) and hence
\begin{equation}
\tau_\perp\ \ge\ \frac{\pi\hbar}{E_{\max}}=\frac{2\pi\hbar\, m}{p_0^{\delta}}\Big(\frac{a}{2\hbar}\Big)^{\alpha}\ ,\qquad F_Q\le\frac{E_{\max}^2}{\hbar^2},
\end{equation}
where $F_Q=4\Delta H^2/\hbar^2$ is the quantum Fisher information for estimating the evolution time. The bound vanishes as $a\to0$ ($E_{\max}\to\infty$), consistent with the unbounded spectra of SQM and FQM.
\end{proposition}
This is a one-line consequence of MT and Popoviciu's inequality. Bounds for spectra of finite width that are dual to ML, and sharper bounded-bandwidth results, are given in \cite{Ness2022,Allan2021}. The bound is asymptotically sharp: an equal-weight superposition of narrow wavepackets at the band edges ($p\simeq0$ and $p\simeq\pi\hbar/a$, with $K=0$ and $K=E_{\max}$) has $\Delta H\to E_{\max}/2$ and $\ave H\to E_{\max}/2$, and saturates MT and ML simultaneously at $\tau_\perp=\pi\hbar/E_{\max}$.

\begin{proposition}[Gaussian states, free evolution]
\label{prop:gaussian}
Let $|\tilde\psi|^2$ be Gaussian with $\ave p=0$, $\ave{p^2}=\sigma^2$, and let $m_\beta\coloneqq2^{\beta/2}\Gamma(\frac{\beta+1}2)/\sqrt\pi$ so that $\ave{|p|^\beta}=m_\beta\sigma^\beta$. With $D_\alpha=p_0^\delta/2m$:
\begin{align}
a=0:&\quad \ave{\hK}=D_\alpha m_\alpha\sigma^\alpha,\qquad \Delta K^2=D_\alpha^2\sigma^{2\alpha}\big(m_{2\alpha}-m_\alpha^2\big),\\
a>0:&\quad \Delta H^2=\Delta K_0^{2}\Big[1-c(\alpha)\frac{a^2\sigma^2}{\hbar^2}\Big]+O(a^4),\\
&\quad \ave{\hK}=\ave{\hK}_0\Big[1-\frac{\alpha\,m_{\alpha+2}}{24\,m_\alpha}\frac{a^2\sigma^2}{\hbar^2}\Big]+O(a^4),
\end{align}
where
\begin{equation}
c(\alpha)=\frac{\alpha}{12}\,\frac{m_{2\alpha+2}-m_\alpha m_{\alpha+2}}{m_{2\alpha}-m_\alpha^2}.
\end{equation}
The orthogonal-state bounds at $a=0$ are $\tau_{\rm MT}=\pi\hbar/2\Delta K$ and $\tau_{\rm ML}=\pi\hbar/2\ave{\hK}$ (here $E_0=0$), i.e.
\begin{equation}
\begin{aligned}
\tau_{\rm MT}&=\frac{\pi\hbar}{2D_\alpha\sigma^\alpha\sqrt{m_{2\alpha}-m_\alpha^2}},\qquad
\tau_{\rm ML}=\frac{\pi\hbar}{2D_\alpha m_\alpha\sigma^\alpha},\\
\frac{\tau_{\rm ML}}{\tau_{\rm MT}}&=\frac{\sqrt{m_{2\alpha}-m_\alpha^2}}{m_\alpha}.
\end{aligned}
\end{equation}
ML is therefore the tighter (larger) bound iff $m_{2\alpha}>2m_\alpha^2$. This holds for $\alpha>\alpha_*\simeq1.383$, while MT is tighter for $\alpha<\alpha_*$; the crossover is independent of $\sigma$, $p_0$ and $m$. At $\alpha=2$ one has $\tau_{\rm ML}=\pi m\hbar/\sigma^2$ and $\tau_{\rm MT}=\pi m\hbar/(\sqrt2\sigma^2)$, so ML exceeds MT by $\sqrt2$.

At $\alpha=2$ and $a>0$: $c=1$, $\Delta H=\frac{\sigma^2}{\sqrt2\,m}\big(1-\frac{a^2\sigma^2}{2\hbar^2}\big)$, and the MT bound is
\begin{equation}
\tau_\perp\ \ge\ \frac{\pi m\hbar}{\sqrt2\,\sigma^2}\Big(1+\frac{a^2\sigma^2}{2\hbar^2}\Big)+O(a^4).
\end{equation}
For general $\alpha$ at $a=0$, the ratio to the $\alpha=2$ value is $\Delta K_\alpha/\Delta K_2=(p_0/\sigma)^{\delta}\sqrt{(m_{2\alpha}-m_\alpha^2)/2}$. The square-root factor is below $1$ for $\alpha<2$. Define the crossover width
\begin{equation}
\sigma_*(\alpha)=p_0\Big[\tfrac12\big(m_{2\alpha}-m_\alpha^2\big)\Big]^{1/(2\delta)},\qquad \sigma_*(1.5)\simeq0.428\,p_0 .
\end{equation}
Then $\Delta K_\alpha<\Delta K_2$ for $\sigma>\sigma_*$, so that fractionality \emph{raises} $\tau_{\rm MT}$ relative to $\alpha=2$ for $\sigma>\sigma_*$ and \emph{lowers} it for $\sigma<\sigma_*$. (Example: $p_0=1$, $\alpha=1.5$, $\sigma=2$ gives $\tau_{\rm MT}(1.5)=1.20>\tau_{\rm MT}(2)=0.56$ in units $m\hbar/p_0^2$.)
\end{proposition}
Thus, to this order and for Gaussian states, the lattice \emph{reduces} $\Delta H$ and raises the MT bound ($c(\alpha)\in[0.547,1]$ for $\alpha\in[1.3,2]$, monotonically increasing in $\alpha$); no acceleration by the lattice is established. The lattice shift of $\Delta H$ is always negative, whereas the fractional shift is positive only for $\sigma<\sigma_*$ (and negative for $\sigma>\sigma_*$). Hence the two shifts can cancel along a curve in the $(a,\alpha)$ plane only for $\sigma<\sigma_*(\alpha)$; for $\sigma>\sigma_*$ both lower $\Delta H$ and no cancellation occurs. Even when they cancel, this only equates the MT bounds and implies nothing further about the dynamics. Derivations and numerical checks (Fig.~\ref{fig:dH}) are in Appendices~\ref{app:qsl} and \ref{app:numerics}.

\begin{figure*}[htbp]
\centering
\includegraphics[width=0.95\textwidth]{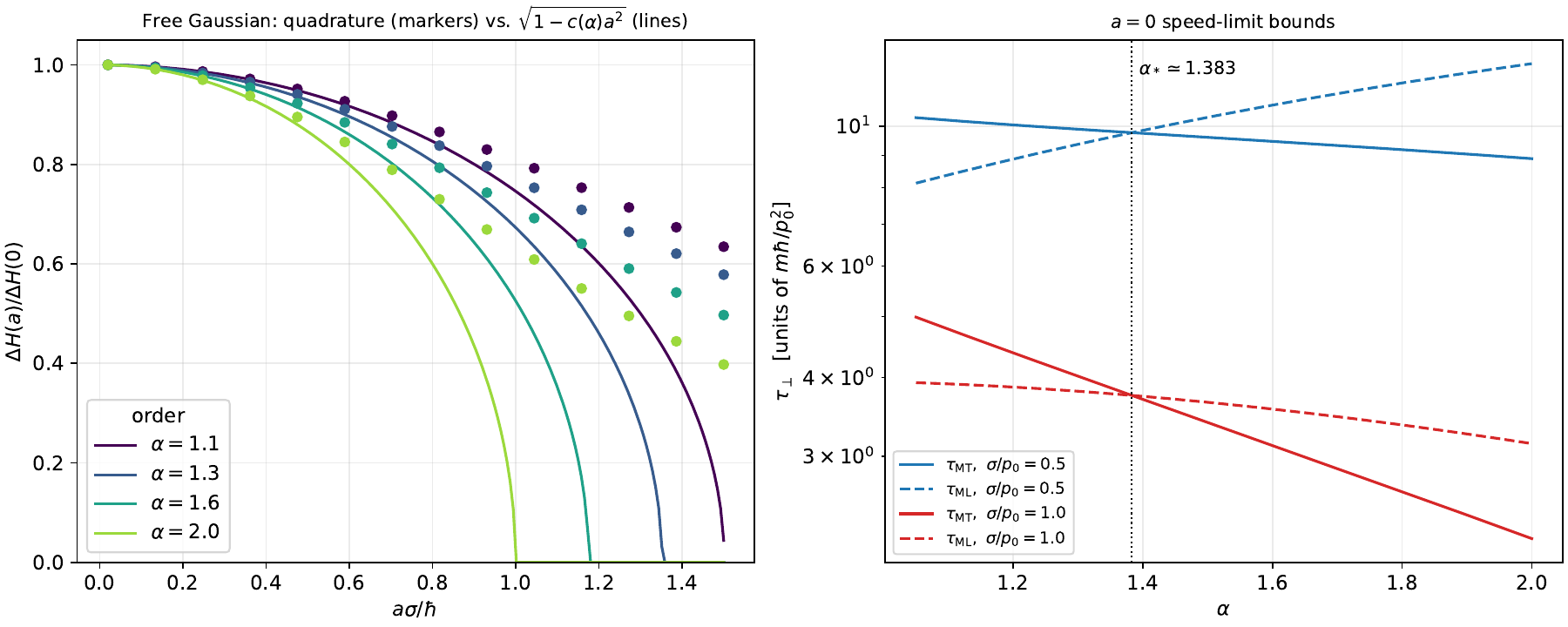}
\caption{Left: $\Delta H(a)/\Delta H(0)$ for a free Gaussian state ($\hbar=m=\sigma=p_0=1$) by direct numerical quadrature (markers) and the $O(a^2)$ result $\sqrt{1-c(\alpha)a^2}$ of Proposition~\ref{prop:gaussian} (lines), for $\alpha=1.1,1.3,1.6,2$, over an extended range of $a\sigma/\hbar$ that displays the $O(a^4)$ deviation. Right: the $a=0$ bounds $\tau_{\rm MT}$ (solid) and $\tau_{\rm ML}$ (dashed) in units of $m\hbar/p_0^2$ versus $\alpha$ for $\sigma/p_0=0.5,1$; the two cross at $\alpha_*\simeq1.383$ independently of $\sigma/p_0$. Quadrature: adaptive Gauss--Kronrod, absolute tolerance at the library default, $|p|\le12\sigma$.}
\label{fig:dH}
\end{figure*}

\section{Phenomenological signatures}
\label{sec:phenomenology}
\subsection{Continuum model versus a genuine lattice}
\label{sec:lattice_vs_continuum}
The model lives on $L^2(\mathbb R)$ with a periodic symbol, i.e.\ it is a direct integral over quasi-momentum of lattice fibers. The uncertainty relations of Section~\ref{sec:uncertainty} use the \emph{continuum} position operator $\hx=i\,d/dk$. On a true lattice $\ell^2(a\mathbb Z)$ the position operator has discrete spectrum and the $x$--$\Pi$ uncertainty relation is different. Consequently, what transfers to a single lattice fiber (for instance the trapped-ion realization below) is: the spectrum $[0,E_{\max}]$ of the fiber operator, the hopping form \eqref{eq:hopping}, the bandwidth QSL, the Gaussian-type MT/ML evaluations for quasi-momentum distributions, and the intra-fiber hopping dynamics. What does \emph{not} transfer is the $x$--$\Pi$ uncertainty relation of Theorem~\ref{thm:uncertainty}. The word ``lattice'' in the title refers to the lattice-periodic dispersion.

\subsection{Lattice dynamics and Bloch oscillations}
By \eqref{eq:hopping}, $U(t)\psi(x)=\sum_nc_n(t)\psi(x-na)$ with $c_n(t)=\frac a{2\pi}\int_{-\pi/a}^{\pi/a}e^{inak-itK(k)/\hbar}dk$: a continuous-time quantum walk on offsets of $a$ with algebraically decaying hops. For $a\to0$ it approaches the L\'evy propagator \cite{Laskin2000}. The autocorrelation is exactly $A(t)=\big|\int|\tilde\psi(k)|^2e^{-itK(k)/\hbar}dk\big|^2$, with short-time form $A\simeq1-(\Delta H\,t/\hbar)^2$.

\paragraph{Constant force.} Let $\hH=\hK-F\hx$ on $L^2(\mathbb R)$. In the $p$-representation $\hx=i\hbar\,\partial_p$, so $\hH=K(p)-i\hbar F\partial_p$. Multiplication by $\exp[-\frac{i}{\hbar F}\int_0^pK]$ is a unitary map taking $\hH$ to $-i\hbar F\partial_p$, which has purely absolutely continuous spectrum $\mathbb R$. Hence, in the continuum model, there is \emph{no} Wannier--Stark ladder, unlike a true lattice, where the quasi-momentum lives on a circle and the ladder (spacing $Fa$) appears. Nevertheless the dynamics is oscillatory in an exact sense. Heisenberg's equations give $\dot{\hat p}=F$ and $\frac d{dt}\ave{\hx}=\ave{v(\hat p)}$, and $|\tilde\psi(p,t)|^2=|\tilde\psi_0(p-Ft)|^2$. Therefore
\begin{equation}
\ave{\hx}(t)=\ave{\hx}(0)+\frac1F\int dq\,|\tilde\psi_0(q)|^2\big[K(q+Ft)-K(q)\big],
\end{equation}
which is \emph{exactly} periodic with period $T_B=2\pi\hbar/(aF)$ for every state, because $K$ is $2\pi\hbar/a$-periodic (equivalently, $v=K'$ has zero period-average). The peak-to-peak amplitude of $\ave{\hx}(t)$ is bounded by $E_{\max}/F$, with equality attained in the sharp-momentum limit (a single plane wave swept across the band); for other states it is smaller and state dependent. The waveform depends on $\alpha$ through \eqref{eq:velocity}. Standard optical lattices realize $\alpha\simeq2$ \cite{Bloch2008}; observing fractional $\alpha<2$ requires engineered long-range hopping, so the Bloch-oscillation waveform is a signature of fractionality only in such an engineered setting.

\paragraph{Metrology.} The QFI for estimating the evolution time with $\nu$ repetitions gives $\Delta t\ge\hbar/(2\sqrt\nu\,\Delta H)$ \cite{Giovannetti2011,Degen2017}; combined with Proposition~\ref{prop:bandwidth}, $F_Q\le E_{\max}^2/\hbar^2$. We make no claim of Heisenberg scaling in $N$.

\paragraph{Open systems.} Coupling to a Markovian bath is described by the standard GKSL equation \cite{Breuer2002} with $\hH_{a,\alpha}$ as the Hamiltonian. The bandwidth bound concerns unitary evolution and does not directly apply to the open dynamics; we do not pursue open-system speed limits here.

\paragraph{Quantum simulation.} The single-excitation sector of an XY model on an ion chain with couplings $J_{ij}\propto|i-j|^{-\beta}$ realizes a particle with hops $|n|^{-\beta}$; the exponent $\beta$ is tunable over $0<\beta<3$ \cite{Richerme2014,Jurcevic2014}. Matching the asymptotics $|n|^{-(1+\alpha)}$ needs $\beta=1+\alpha\in(2,3]$, within the accessible range but near its edge, and reproducing the full coefficients $\kappa_n$ would require additional engineering (see \cite{Georgescu2014}). A sign issue must be addressed: $\kappa_n<0$ for $n\ne0$ (and $\kappa_0>0$), whereas trapped-ion couplings $J_{ij}$ are typically of one sign. Depending on the sign convention, the target is realized directly or by implementing $-\hK+\text{const}$ instead; this is equivalent for the autocorrelation $A(t)$ (since $|\tilde\psi|^2$ is real, $K\to-K$ complex-conjugates the amplitude), for $\Delta H$, and for the MT and bandwidth bounds. We offer this as an experimental route, not as a performed measurement. Since a single ion-chain fiber is a true lattice, only the fiber-level results listed in Section~\ref{sec:lattice_vs_continuum} apply.

\section{Conclusion}
We studied fractional kinetics with a lattice-periodic dispersion: a bounded Fourier multiplier on $L^2(\mathbb R)$ with purely absolutely continuous spectrum $[0,E_{\max}]$. The main new technical result is a domain-controlled commutator for the singular multiplier $M$ at $\alpha\in(1,2)$ and the resulting uncertainty relation $\Delta x\Delta\Pi\ge\frac\hbar2|\ave M|$, whose $\alpha=2$ case is the known lattice relation \cite{Jizba2010}, with a joint expansion and a statement of its range of validity. For speed limits we gave closed-form Gaussian MT and ML evaluations, the MT/ML crossover at $\alpha_*\simeq1.383$, and the crossover width $\sigma_*$ governing the sign of the fractional effect; the bandwidth bound is a simple consequence of MT and is weaker than known bounded-spectrum results. To leading order the lattice lowers $\Delta H$ and the uncertainty bound, so it does not accelerate Gaussian evolution. Open questions include many-body generalizations, exact treatment of orthogonality time for non-Gaussian states, the case $\alpha\le1$, and engineering of the hopping coefficients $\kappa_n$.

\appendix

\section{Commutator details}
\label{app:commutator}
Write $F(k)=\Pi_{a,\alpha}(\hbar k)$. Away from the zeros $k_n=2\pi n/a$ of $s_a$, $F$ is smooth and
\[
F'(k)=\hbar\,\partial_p\Pi=\hbar\,\frac\alpha2\Big|\frac{s_a}{p_0}\Big|^{\alpha/2-1}\cos\frac{ak}2=\hbar M(\hbar k),
\]
using $\partial_p\,\mathrm{sgn}(s)|s|^{\alpha/2}=\frac\alpha2|s|^{\alpha/2-1}s'$ and $s_a'(p)=\cos\theta$. Near $k_n$, $F\sim\mathrm{sgn}(k-k_n)|k-k_n|^{\alpha/2}$, which is absolutely continuous with derivative $\sim|k-k_n|^{\alpha/2-1}$; no Dirac-delta terms arise. The derivative is in $L^2_{\rm loc}$ iff $\alpha-2>-1$, i.e.\ $\alpha>1$. The remaining steps are in the proof of Lemma~\ref{lemma:commutator}. In the $k$-representation $\hx=i\,d/dk$ (consistently with $\mathcal F[x\psi]=i\tilde\psi'$), so no factor of $\hbar$ is lost between $k$- and $p$-derivatives: $F'=\hbar\,\partial_p\Pi$ gives $[\hx,\hat\Pi]=i\hbar M(\hp)$.

\section{Weak-lattice expansion}
\label{app:expansion}
With $|s_a|=|p|\,\big|\frac{\sin\theta}\theta\big|$, $\theta=ap/2\hbar$, and $\frac\alpha2=1-\frac\delta2$,
\begin{align*}
M&=\Big(1-\frac\delta2\Big)\Big|\frac p{p_0}\Big|^{-\delta/2}\Big|\frac{\sin\theta}\theta\Big|^{-\delta/2}\cos\theta\\
&=\cos\theta\,\exp\!\Big[-\frac\delta2\Big(1+L+\ln\Big|\frac{\sin\theta}\theta\Big|\Big)\Big]+O(\delta^2),
\end{align*}
which is the first statement of Proposition~\ref{prop:expansion}. Using $\cos\theta=1-\theta^2/2$, $\ln|\sin\theta/\theta|=-\theta^2/6$, and $\theta^2=a^2p^2/4\hbar^2$,
\[
M=1-\frac{\theta^2}2-\frac\delta2(1+L)+\delta\theta^2\Big(\frac13+\frac L4\Big)+O(\theta^4,\delta^2),
\]
and $\delta\theta^2(\frac13+\frac L4)=\frac{\delta a^2p^2}{\hbar^2}(\frac1{12}+\frac L{16})$, giving \eqref{eq:uncertainty_expanded}. The conditions under which term-by-term integration is justified are stated after Proposition~\ref{prop:expansion}. For a Gaussian, $\ave{\ln|p|}=\ln\sigma+\frac12(\psi(\frac12)+\ln2)$.

\section{QSL expansions}
\label{app:qsl}
For $a\to0$, $K\simeq D_\alpha|p|^\alpha\big[1-\frac{\alpha a^2p^2}{24\hbar^2}\big]$ since $|s_a|^\alpha=|p|^\alpha(1-\theta^2/6)^\alpha$. With Gaussian moments $\ave{|p|^\beta}=m_\beta\sigma^\beta$,
\[
\Delta K^2=D_\alpha^2\Big[\sigma^{2\alpha}(m_{2\alpha}-m_\alpha^2)-\frac{\alpha a^2}{12\hbar^2}\sigma^{2\alpha+2}(m_{2\alpha+2}-m_\alpha m_{\alpha+2})\Big],
\]
which gives $c(\alpha)$. At $\alpha=2$: $m_2=1$, $m_4=3$, $m_6=15$, so $c=\frac16\cdot\frac{15-3}{3-1}=1$ and $\Delta H^2=\frac{\sigma^4}{2m^2}(1-a^2\sigma^2/\hbar^2)$, giving $\tau_\perp\ge\pi m\hbar/(\sqrt2\sigma^2)$ at $a=0$. The ML/MT crossover follows from $\tau_{\rm ML}/\tau_{\rm MT}=\sqrt{m_{2\alpha}-m_\alpha^2}/m_\alpha$, which equals $1$ iff $m_{2\alpha}=2m_\alpha^2$; numerically this holds at $\alpha_*=1.38297\ldots$ and the ratio is $<1$ for $\alpha<\alpha_*$, $>1$ above. The crossover width follows from setting $\Delta K_\alpha/\Delta K_2=1$ in the ratio given in Proposition~\ref{prop:gaussian}.

\section{Numerical check}
\label{app:numerics}
Figure~\ref{fig:dH} was obtained by numerical quadrature of $\ave K$ and $\ave{K^2}$ for a Gaussian $|\tilde\psi|^2$ with $\hbar=m=\sigma=p_0=1$, integrating over $|p|\le12\sigma$ (adaptive Gauss--Kronrod with a break point at $p=0$), for $a\sigma/\hbar\in[0.02,1.5]$. At $a=0.05$, the numerical ratio $\Delta H^2(a)/\Delta H^2(0)$ differs from $1-c(\alpha)a^2$ at the $10^{-6}$ level for $\alpha\in\{1.5,1.8\}$ (e.g.\ $0.998347$ versus $0.998345$ at $\alpha=1.5$), and the relative change of $\ave K$ agrees with the closed form. The values $c(1.3)=0.547$, $c(1.5)=0.662$, $c(1.8)=0.856$, $c(2)=1$ and $\alpha_*=1.383$, $\sigma_*(1.5)=0.428\,p_0$ were evaluated from the closed forms and confirmed by quadrature. The expansion of $M$ in Appendix~\ref{app:expansion} was checked pointwise (relative error $\lesssim10^{-3}$ at $\theta\lesssim0.02$, $\delta\le0.05$).


\end{document}